\documentclass[11pt]{article}

\usepackage{amssymb,amsbsy,amsmath,amsfonts,amssymb,amscd,amsthm}
\usepackage[ruled,vlined]{algorithm2e}
\usepackage[utf8]{inputenc}
\usepackage{graphicx}
\usepackage{complexity}
\usepackage{fullpage}
\usepackage{todo}
\usepackage{hyperref}                   
\usepackage{authblk}

\title{Enumerating models of DNF faster: breaking the dependency on the formula size }
\author[1]{Florent Capelli}
\affil[1]{Universit\'{e} de Lille, CRIStAL Laboratory, France, florent.capelli@univ-lille.fr}
\author[2]{Yann Strozecki}
\affil[2]{Universit\'{e} de Versailles Saint-Quentin-en-Yvelines, DAVID Laboratory, France}

\newtheorem{theorem}{Theorem}
\newtheorem{lemma}[theorem]{Lemma}
\newtheorem{proposition}[theorem]{Proposition}

\theoremstyle{definition}

\newcommand{\strongpdelay}{{strong polynomial delay}}
\newcommand{\sat}{\mathsf{sat}}
\newcommand{\var}{\mathsf{var}}

\newcommand{\Output}{\mathtt{Output}}

\newcommand{\enumDNF}{{\sc EnumDNF}}
\newcommand{\enum}[1]{\textsc{Enum}\smash{\cdot}#1}
\newcommand{\one}{\mathbf{1}}
\newcommand{\zero}{\mathbf{0}}

\begin{document}
\maketitle

\begin{abstract}
In this article, we study the problem of enumerating the models of DNF formulas. The aim is to provide enumeration algorithms with a delay that depends polynomially on the size of each model and not on the size of the formula, which can be exponentially larger. We succeed for two subclasses of DNF formulas: we provide a constant delay algorithm for $k$-DNF with fixed $k$ by an appropriate amortization method and we give a quadratic delay algorithm for monotone formulas. We then focus on the \emph{average delay} of enumeration algorithms and show how to obtain a sublinear delay in the formula size. 

\end{abstract}

\newpage



%
%
%
%

\section{Introduction}

An enumeration problem is the task of listing a set of elements without redundancies, usually corresponding to the solutions of a search problem, such as enumerating the spanning trees of a graph or the satisfying assignments of a formula. One way of measuring the complexity of an enumeration algorithm is the \emph{total time} needed to compute all solutions. When the total time depends both on the \emph{input and output}, an algorithm is called \emph{output sensitive}. It is considered tractable and said to be \emph{output polynomial} when it can be solved in polynomial time in the size of the input and the output. 

Output sensitivity is relevant when \emph{all} elements of a set must be
generated, for instance to build a library of interesting objects to be studied
by experts, as it is done in biology, chemistry or network
analytics~\cite{barth2015efficient,andrade2016enumeration,bohmova2018computing}.
However, when the output is large with regard to the input, output polynomiality
is not enough to capture tractability. Indeed, if one wants only a good solution
or some statistic on the set of solutions, it can be enough to generate only a
fraction of the solutions. A good algorithm for this purpose must guarantee that
the more time it has, the more solutions it generates. To measure the efficiency
of such an algorithm, we need to evaluate the {\em delay} between two
consecutive solutions. A first guarantee that we can expect is to have a good
{\em average delay} (sometimes referred to as {\em amortized delay} or {\em time
  per solutions}), that is, to measure the total time divided by the number of
solutions. There are many enumeration algorithms which are in constant amortized
time or CAT, usually for the generation of combinatorial objects, such as the
unrooted trees of a given size~\cite{wright1986constant}, the linear extensions
of a partial order~\cite{pruesse1994generating} or the integers by Gray
code~\cite{knuth1997art}. Uno also proposed in~\cite{uno2015constant} a general
method to obtain constant amortized time algorithms, which can be applied, for
instance, to find the matchings or the spanning trees of a graph.

However, when one wants to process a set in a streaming fashion such as the
answers of a database query, it may be interesting to guarantee that we have a
good delay between the output of two new solutions, usually by bounding it by a
polynomial in the input size. We refer to such algorithms as \emph{polynomial
  delay} algorithms. Many problems admit such algorithms, e.g. enumeration of
the cycles of a graph~\cite{read1975bounds}, the satisfying assignments of
some tractable variants of $\SAT$~\cite{creignou1997generating} or the spanning trees and
connected induced subgraphs of a graph~\cite{avis1996reverse}. All polynomial
delay algorithms are based on few methods such as \emph{backtrack search} or
\emph{reverse search}, see~\cite{mary2013enumeration} for a survey.

For some problems, the size of the input may be much larger than the size of
each solution, which makes polynomial delay an unsatisfactory measure of
efficiency. In that case, algorithms whose delay depends on the size of a single
solution are naturally more interesting than polynomial delay or output
polynomial algorithms. We say that an algorithm is in \emph{\strongpdelay} when
the delay between two consecutive solutions is polynomial in the size of the
last solution. To make this notion robust and more relevant, a precomputation
step is allowed before the enumeration, in time polynomial in the input size, so
that the algorithm can read the whole input and set up useful data structures.
Observe that the notion of {\strongpdelay} is also well suited for the
enumeration of infinite sets where the size of the solutions may grow
arbitrarily as in~\cite{florencio2015naive}.

There are few examples of {\strongpdelay} algorithms in the literature. A folklore example is the enumeration of the paths in a DAG, which is in delay linear in the size of the generated path. More complex problem can then be reduced to generating paths in a DAG, such as enumerating the minimal dominating sets in restricted classes of graphs~\cite{golovach2018output}.

\emph{Constant delay} algorithms also naturally fall into  {\strongpdelay} algorithms and a whole line of research is dedicated to design such algorithm for enumerating models of first order queries for restricted classes of structures~\cite{SchweikardtSV18} (see also the survey~\cite{segoufin2013enumerating}). However, while these algorithms are called constant delay because their delay does not depend on the database size, it often depends more than exponentially on the size of the solutions. Other examples naturally arise from logic such as the enumeration of assignments of MSO queries over trees or bounded width graphs~\cite{bagan2006mso,courcelle2009linear}.

In this paper, we focus on the problem of enumerating the models of DNF formulas
that we denote by \enumDNF. More precisely, we would like to understand whether
{\enumDNF} is in {\strongpdelay} or not. The complexity of counting the models
of a DNF is well understood: while exactly counting the number of models of a
DNF is a canonical $\#\P$-complete problem, it is known that approximating this
number can be approximated in randomized polynomial
time~\cite{karp1983monte,meel2018not}. This result exploits the fact that the
structure of the models of a DNF is very simple: it is the \emph{union} of the
models of its terms. This structure also suggests that {\enumDNF} is easy to
solve and it is folklore that the models of a DNF $D$ can be enumerated with a
delay linear in the size of $D$ (see Proposition~\ref{prop:classical_flashlight}).
This algorithm has polynomial delay but not a {\strongpdelay} since the number
of terms of a DNF may be super-polynomial in the number of variables.

The main difficulty to obtain a {\strongpdelay} algorithm for the problem
{\enumDNF} is to deal with the fact that the models are repeated many times
since \emph{the union is not disjoint}. Solution repetitions because of non
disjoint union is a common problem in enumeration and this issue appears in its
simplest form when solving {\enumDNF}. We hope that understanding finely the
complexity of this problem and giving better algorithm to solve it will shed
some light on more general problems. In fact, it turns out that the algorithms
we propose in this article may be directly applied to get better algorithms for
generating union of sets~\cite{mary2016efficient,dmtcs:5549} and models of
existential second order queries~\cite{durand2011enumeration}. A recent result
by Amarilli et al.~\cite{AmarilliBJM17} naturally illustrates the link between
{\strongpdelay} and disjoint unions: they give a {\strongpdelay} algorithm to
enumerate the models of Boolean circuits known as structured d-DNNF used in
knowledge compilation. These circuits may be seen as way of generating a set of
solutions by using only Cartesian products and \emph{disjoint} unions of ground
solutions. Many known enumeration algorithm such as the enumeration of the
models of an MSO formula on bounded treewidth databases can be reduced to the
enumeration of the solutions of such circuits. This suggests that the boundary
between {\strongpdelay} and polynomial delay may be related to the difference
between union and disjoint union of solutions.


We thus think that {\enumDNF} is a reasonable candidate to separate
{\strongpdelay} from polynomial delay and our working hypothesis are the
following two conjectures:

\begin{enumerate}
 \item \emph{DNF Enumeration Conjecture:} {\enumDNF} is not in {\strongpdelay}.
 \item \emph{Strong DNF Enumeration Conjecture:} There is no algorithm solving
   the problem {\enumDNF} in delay $o(m)$ where $m$ is the number of terms of
   the DNF.
\end{enumerate}

In this paper, we refute stronger forms of these conjectures on restricted
formulas such as monotone DNF and $k$-DNF, or by considering average delay, see
Table~\ref{tab:results} for a summary of the results presented in this paper.
The general case is still open however but we hope to make some progress by
reducing the conjectures to strong hypotheses such as SETH, as it was done
in~\cite{capelli2019incremental} to prove a strict hierarchy in incremental
delay enumeration algorithms.

\begin{table}
  \centering
  \begin{tabular}{|l|l|l|}
    \hline
    {\bf Class} & {\bf Delay} & {\bf Space} \\ \hline
    DNF & $O(\|D\|)$ (Proposition~\ref{prop:classical_flashlight}) & Polynomial \\ \hline
    $(\star)$ DNF & $O(nm^{1-\log_3(2)})$ average delay (Theorem~\ref{thm:best_average_delay}) & Polynomial \\ \hline
    $(\star)$ $k$-DNF  & $2^{O(k)}$ (Theorem~\ref{thm:kdnf}) & Polynomial \\ \hline
    $(\star)$ Monotone DNF & $O(n^2)$ (Theorem~\ref{thm:montonednf}) & Exponential \\ \hline
    $(\star)$ Monotone DNF & $O(\log(nm))$ average delay (Theorem~\ref{th:amortized_improved}) & Polynomial \\ \hline
  \end{tabular}
  \caption{Overview of the results. In this table, $D$ is a DNF, $n$ its number of variables and $m$ its number of terms. New contributions are annotated with $(\star)$.}
  \label{tab:results}
\end{table}

\paragraph*{Organization of the paper}

The paper is organized as follows. We first introduce basic notions on formulas and enumeration complexity then present tries and Gray code in Section~\ref{sec:def}. We show in Section~\ref{sec:classical_alg} how to adapt three generic methods to generate the models of a DNF, the best having a linear delay in the formula size. In Section~\ref{sec:kdnf}, we give a backtrack search algorithm, using a branching scheme which supports a good amortization, which is in constant delay for $k$-DNF. In Section~\ref{sec:average}, we give another backtrack search algorithm, whose average delay is sublinear using a lemma relating the number of models of a DNF with its number of terms. Finally, in Section~\ref{sec:monotone}, we provide a {\strongpdelay} for monotone DNF formulas using exponential memory and we specialize and adapt the algorithm of the previous section to obtain a logarithmic delay for monotone DNF formulas with polynomial memory.

\section{Definitions and notations}
\label{sec:def}

\paragraph*{Terms and DNF-formulas}

Let $X$ be a set of variables and let $n$ and be the size of $X$. We fix some arbitrary order
on $X$ and write $X = \{x_1,\dots, x_n\}$. A {\em literal} $\ell$ is either a variable $x \in X$ or the negation of a variable $\neg x$ for some $x \in X$. A {\em term} $T$ is a finite set of literals such that every two literals in $T$ have a different variable. A {\em Disjunctive Normal Form formula}, DNF for short, is a finite set of terms.
Given a literal $\ell$, we denote its underlying variable by $\var(\ell)$. We extend this notation to terms by denoting $\var(T) := \bigcup_{\ell \in T} \var(\ell)$ for a term $T$ and to DNF by denoting $\var(D) := \bigcup_{T \in D} \var(T)$ for a DNF $D$.

Given an assignment $\alpha : X \rightarrow \{0,1\}$, we naturally extend $\alpha$ to literals by defining $\alpha(\neg x) = 1-\alpha(x)$. An assignment $\alpha$ satisfies a term $T$ if for every $\ell \in T$, $\alpha(\ell) = 1$. A {\em model} $\alpha$ of a DNF $D$ is an assignment of $\var(D)$ such that there exists $T \in D$ such that $\alpha$ satisfies $T$. We write $\alpha \models D$ if $\alpha$ is a model of $D$. It is easy to see that given a term $T$, there exists a unique assignment of variables in $\var(T)$ satisfying $T$. We denote this assignment by $\one_T$. 

Given a DNF $D$ on variables $X$, we denote by $\sat(D) = \{\alpha \mid \alpha \models D\}$ the set of models of $D$.
Let $Y \subseteq X$, $\tau : Y \rightarrow \{0,1\}$ and $\sigma : X \rightarrow \{0,1\}$,  we say that $\sigma$ is {\em compatible} with $\tau$, denoted by $\sigma \simeq \tau$, if the restriction of $\sigma$ to $Y$ is equal to $\tau$. We denote by $\sat(D,\tau) = \{\alpha \mid \alpha \models D, \sigma \simeq \tau\}$ the set of models of $D$ compatible with $\tau$. We also denote by $D[\tau]$ the DNF defined as follows: we remove every term $T$ from $D$ such that there exists a literal $\ell \in T$ such that $\tau(\ell) = 0$. For the remaining terms, we remove every literal whose variable is in $Y$. Observe that since we consider DNF to be sets of terms, by definition, $D[\tau]$ has no duplicated terms. It is clear that: $\sat(D,\tau) = \{\tau \cup \alpha \mid \alpha \models D[\tau] \}$. The {\em size} of a DNF $D$ is denoted by $\|D\|$ and is equal to $\sum_{T \in D} |T|$.

\paragraph*{Enumeration complexity} 

Let $\Sigma$ be a finite alphabet and $\Sigma^*$ be the set of finite words built on $\Sigma$.
Let $A\subseteq \Sigma^{*}\times\Sigma^{*}$ be a binary predicate, we write $A(x)$ for the set of $y$ such that $A(x,y)$ holds. The enumeration problem $\enum{A}$ is the function which associates $A(x)$ to $x$.

The computational model is the random access machine model (RAM) with addition, subtraction and multiplication as its basic arithmetic operations and an operation $\Output(i,j)$ which outputs the concatenation of the values of registers $R_i, R_{i+1}, \dots, R_j$. We assume that all operations are in constant time except the arithmetic instructions which are in time linear in the size of their inputs.

A RAM machine solves $\enum{A}$ if, on every input $x \in \Sigma^{*}$, it produces a sequence $ y_{1}, \dots, y_{n}$ such that $ A(x) = \left\lbrace y_{1}, \dots, y_{n} \right\rbrace $ and for all $i\neq j,\, y_{i} \neq y_{j}$.
The space used by the machine at a given step is the sum of the number of bits required to store the integers in its registers. 

The delay of a RAM machine which outputs the sequence $\left\lbrace y_{1}, \dots, y_{n} \right\rbrace $ is the 
maximum over all $i\leq n$ of the time the machine uses between the generation of $y_i$ and $y_{i+1}$ and between the generation of $y_n$ and when the machine stops. 
Note that we allow a \emph{precomputation phase} before the machine starts to enumerate solutions,
which can be in time polynomial in the size of the input.

Remark that this formalism allows for the enumeration of a solution of unbounded size in constant time.
The later is required to have algorithms with constant delay.  This model is reasonable if we store the deltas between solutions rather than the solutions. It is also relevant when instead of storing solutions
we only need to do a constant time operation on each solution such as counting the solutions or evaluating some measure which depends only on the constant amount of change between two consecutive solutions.

\paragraph*{Trie}\label{par:trie}

A {\em trie} is a data structure used to represent a set of words on an alphabet $M$ which supports efficient insertion and deletion. We refer the reader to~\cite{fredkin1960trie,knuth1997art} for more details. We use them to either store a formula, i.e. the set of the terms of a DNF, or to represent a set of models of a DNF. 

The trie is a $M$-ary tree, each of its inner nodes is labeled by a letter of the alphabet. A leaf represents a word: the sequence of labels on the path from the root to the leaf. The trie represents the set of words represented by its leaves. 

A model $\alpha$ of the DNF is represented in a trie by the sequence of its
values: $\alpha(x_1)\alpha(x_2)\dots\alpha(x_n)$. The trie we use to store these
models are binary trees, since the labels are in $\{0,1\}$.

A term $T$ of a DNF is represented by the sequence in ascending order of its
literals using the order defined as follows: for all $i$, $\bar{x_i} < x_i$ and
$x_i < \bar{x}_{i+1}$.

In this paper, the operation that we need to perform on tries are: searching for an element, inserting a new element and suppress an element. The complexity of these operations and the space used by the 
trie depends on the way the data structure is implemented. We claim that we can implement tries so that these three operations run in time $O(n)$ with no space overhead where $n$ is the size of the longest word stored in the tree. As we will use tries to store models of DNF formulas or terms, $n$ will always be bounded by the number of variables of the DNF formula. To do this, we store the children of each node in a sorted linked list.

  The size of a term, denoted by $k$, may be small with regard to the size of
  the alphabet, denoted by $n$. When storing terms in a trie, using an array to represent the children of each node, we can achieve a complexity of $O(k)$ for all operations. Note that, to obtain this complexity we must either assume an infinite supply of \emph{initialized memory} or use a lazy intialization technique (see~\cite{mehlhorn2013data} 2, Section III.8.1). The space usage can be multiplied by a factor $n$ because we use a size $n$ array to represent the children relation 
  which can be empty. There are better way to implement tries with a low space overhead but we do not try to optimize this parameter in this article.

%
%

\paragraph*{Gray Code}

Gray codes are an efficient way to enumerate the integers between $0$ and $2^n
-1$ written in binary, or equivalently the subsets of a set of size $n$. They
enjoy two important properties: the Hamming distance of two consecutive elements in the Gray enumeration order is one and each new element is produced in constant time using only additional $O(\log(n)^2)$ space (see~\cite{knuth2011combinatorial}). In other words, this enumeration algorithm has constant delay. We can generate all models of a term in constant delay, using Gray code and an additional array which contains the indexes of the free variables of the term.

\begin{proposition} \label{prop:Gray}
 The models of a term $T$ on variables $X$ can be enumerated in constant delay.
\end{proposition}

\begin{proof}

  The models of a term $T$ on variables $X$ are exactly the assignments of the form $\one_T \cup \tau$ for any $\tau : X \setminus \var(T) \rightarrow \{0,1\}$. Enumerating the models of $T$ intuitively boils down to enumerating all assignments on variables ${X \setminus \var(T)}$. We use a Gray code to do this in constant time.
  
  More precisely, we represent a model of $T$  in $n$ registers $R_1,\dots,R_n$ where $R_i$ holds the value of $x_i$ in this model. We initialize the registers such that $R_i = 1$ if $x_i \in T$ and $R_i = 0$ otherwise, that is, if $\neg x_i \in T$ or $x_i \notin \var(T)$.
  
Now let $k = |X \setminus \var(T)|$ and $\sigma: \{1,\dots,k\} \rightarrow \{1,\dots, n\}$ be such that $X \setminus \var(T) = \{x_{\sigma(1)}, \dots, x_{\sigma(k)}\}$ with $\sigma(1) < \dots < \sigma(k)$. We start by storing the values of $\sigma$ in an array of size $k$. We then execute $\Output(1,n)$ which outputs the first model of $T$. After that, we run a Gray code enumeration on the subsets of a set of size $k$. For each new element generated, the Gray code algorithm flips a bit at some position $i$. We thus switch the bit of $R_{\sigma(i)}$ and call $\Output(1,n)$. This generates all solutions of $T$ since all possible values of the variables not in $T$ are set, while the other have the only allowed value by $T$. The delay between two solutions is constant since we only look the value of $\sigma(i)$ in an array and switch the value of a register between two outputs.
\end{proof}

\section{Classical enumeration algorithms}\label{sec:classical_alg}

In this section, we solve {\enumDNF} thanks to three generic enumeration
methods. The best one has a linear delay in the size of the instance, and we
study in later sections several restrictions to obtain a delay polynomial in the
size of a solution. This algorithms presented in this section are not
{\strongpdelay} and most of them are folklore or adapted from other algorithms.
We recall them in this paper as some of our improved algorithms rely on
improvement of these algorithms.

\subsection{Union of terms}

The first two methods are based on the fact that the models of a DNF formula are the union of the models 
of its terms. The main difficulty is to avoid repetitions. We first use a method to enumerate the union of sets of elements which preserves polynomial delay (see~\cite{strozecki2010enumeration}). It relies on a priority rule between the sets to avoid repetitions, as recalled in the proof of the following proposition.

\begin{proposition}[Adapted from proposition $2.38$ and $2.40$ in~\cite{strozecki2010enumeration}]\label{prop:enum_union}
 The models of a DNF formula $D$ with $m$ terms can be enumerated with delay $O(m\|D\|)$.
\end{proposition}
\begin{proof}
 As explained in Proposition~\ref{prop:Gray}, given a term $C \in D$, $\sat(C)$ can be enumerated in constant delay using Gray code enumeration. The terms of $D$ are indexed from $C_1$ to $C_m$ in an arbitrary order. During the algorithm, the state of the enumeration by Gray Code for each term is maintained so that we can query in constant time the next model of a given term.
 At each step, the algorithm does a loop from $C_1$ to $C_m$. For a term $C_i$ the next model is generated and the algorithm tests whether it is a model of some $C_j$ with $j >i$. If not, it is output. If a term has no more models we skip it.
 
 By this method, we guarantee that each model is output when generated by the term of largest index it satisfies, hence all models are generated and without repetitions. Moreover, at each step of the algorithm, the model given by the last term which has still models will be output, therefore the delay is bounded by the time to execute one step of the algorithm. 
 
 The cost of the generating new models at each step is bounded by $O(m)$ since each solution is produced in $O(1)$. The cost of testing whether a model satisfies some term of larger index is bounded by $\|D\|$ and it is done at most $m$ times before outputting a solution which implies that the delay is $O(m\|D\|)$.
\end{proof}

By improving the way we test whether a model of a term is the model of any term of larger index, we can drop the delay 
to $O(m^2)$. In fact, by generating the solutions of each term in the same order, we can avoid completely the redundancy test and get a better delay. The following algorithm merges several ordered arrays which are generated dynamically by enumeration procedures.

\begin{proposition}[Adapted from proposition $2.41$ in~\cite{strozecki2010enumeration}]\label{prop:ordered_union}
 The models of a DNF formula $D$ with $m$ terms  and $n$ variables can be enumerated with delay $O(mn)$.
\end{proposition}
\begin{proof}
 As in the previous algorithm, we run a simple enumeration algorithm on each term and we maintain 
 their states so that we can easily query the next model of a term. We chose to enumerate the models of the terms in lexicographic ascending order (for some arbitrary order of the variables), which can be done with delay $O(n)$ for each term.
 
 The first model of each term which has not yet been output is stored in a
 trie; if all models of a term have been output then nothing is stored for
 this term. Moreover, for each model $\alpha$ in the trie, we maintain the list
 of terms $L$ from which it has been generated. This can be achieved by labeling
 the leaf corresponding to $\alpha$ in the trie with $L$.

  At each step of the algorithm, the smallest model $\alpha$ is found in the trie, then output and removed from the trie all in time $O(n)$. Then we use the list of terms which had $\alpha$ as a model, to generate for each of them their next model and add it to the trie in time $O(mn)$ since the insertion can be done in time $O(n)$ and the number of new models is bounded by $O(m)$.
 The delay of this algorithm is thus bounded by $O(mn)$. By induction, we prove that at each step the smallest non output model is output, which implies that all models are output without repetitions.
\end{proof}

The average delay of the two previous algorithms is lower than their delays. 
Let $r$ be the average number of times a solution is produced during the algorithm, 
we can replace $m$ by $r$ in the average delay of the previous algorithm. It is possible to prove that $r$
is smaller than $m$ by studying how terms share models, but the complexity gains are very small
and we give an algorithm with a better average delay in Section~\ref{sec:average}.

\subsection{Flashlight method}

We present a classical enumeration method called the \emph{Backtrack Search} or
sometimes the \emph{Flashlight
  Method}~\cite{read1975bounds,strozecki2013enumerating}, used in particular to
solve auto-reducible problems. We describe the method in the context of the
generation of the models of a DNF formula.

Given a DNF $D$ on variables $x_1,\dots, x_n$, we define a rooted tree $T_D$
whose root is labeled with $\emptyset$ and nodes of depth $k$ are the
assignments $\tau$ over variables $\{x_1,\dots x_k\}$ such that there exists a
model $\sigma$ of $D$ which is compatible with $\tau$. The children of a node
labeled by $\tau$ are the partial assignments $\tau'$ defined over $\{x_1,\dots
x_{k+1}\}$ which are compatible with $\tau$.

The leaves of $T_D$ are the models we want to generate, therefore a depth first traversal visits all
leaves and thus outputs all solutions. Since a path from the root of the tree is of size  $n$, it is enough to be able to find the children of a node in polynomial time to obtain a polynomial delay. Hence the Flashlight Method has a polynomial delay if and only if the  following extension problem is in $\P$:
given $\tau$ over  $\{x_1,\dots x_k\}$  is there $\sigma$ a model of $D$ compatible with $\tau$?

The extension problem for a DNF is very simple to solve: compute the formula $D[\tau]$ and decide whether it is
satisfiable in time $O(\|D\|)$. This yields an enumeration algorithm with delay $O(n\|D\|)$. 
The delay can be improved by using the fact that we solve the extension problem several times on similar instances as it has been done for other problems such as the enumeration of the models of a monotone CNF~\cite{murakami2014efficient} or the union of sets~\cite{mary2016efficient}.

\begin{proposition}\label{prop:classical_flashlight}
 The models of a DNF formula $D$ can be enumerated with delay $O(\|D\|)$.
\end{proposition}
\begin{proof}
  We use the previous algorithm which does a depth first search in $T_D$. When it visits 
 the node  labeled by $\tau$, we need to decide whether $D[\tau]$ is satisfiable which is equivalent to testing whether it has a non falsified term.  To speed-up the flashlight search, we use a data structure to decide quickly this problem and we need to guarantee that it can be updated fast enough when the tree is traversed.
 
 For each term $T$, we store an integer $f_T$ which represents how many literals of the term are falsified by the current partial assignment. A term $T$ is valid when $f_T = 0$. For each literal $l$, we store the list of terms which contain the literal. 
 We also store an integer $c$ for valid terms, which counts how many terms are not falsified by the current partial assignment.
 
 At the beginning of the algorithm, all $f_T$ are set to $0$ and $c = m$. Visiting a child $\tau'$ of $\tau$ corresponds to setting the value of some variable $x_k$. If $x_k$ is set to $0$ then for each term $T$ containing $x_k$, $f_T$ is incremented. The number of terms such that $f_T$ is changed from $0$ to $1$ is subtracted to $c$. The case where $x_k$ is set to $1$
 is analogous (consider that $\bar{x_k}$ is set to $0$). The fact that $D[\tau']$ is satisfiable is equivalent to $c > 0$. Remark that going up the tree when backtracking works exactly as going down, but the variables $f_T$ are decremented and $c$ is incremented instead.
 
 As a consequence, the complexity of the algorithm over a path in the tree is $O(\|D\|)$ since for each term $T$, the variable $f_T$ will be modified at most $|T|$ times. When the algorithm goes down the tree it may first set $x_k$ to $0$ and fail, but then it sets $x_k$ to $1$ and goes down. Hence the cost of going down to a leaf is at most twice the cost of following the path to the leaf. Since, between two output solutions, the algorithm follows one path up and one down, the delay is in $O(\|D\|)$.
 \end{proof}

 This last algorithm has a delay linear in the size of the formula, however the
 size of the formula can be much larger than the size of a model, if the number
 of terms in the DNF is super-polynomial in the number of variables for example.
 In the next sections we will try to reduce or eliminate the dependency of the
 delay in the size of the input either for particular DNF formulas or by
 relaxing the notion of delay.

\section{Enumerating models of $k$-DNF}
\label{sec:kdnf}

A term $T$ is a $k$-term if and only if $|T| \leq k$. A DNF is a $k$-DNF if all its terms are $k$-terms. In this section, we present an algorithm to enumerate the models of a $k$-DNF with a $2^{O(k)}$ delay. The idea is to select a $k$-term and use its $2^{n-k}$ models to amortize more costly operations. More precisely, we prove the following:

\begin{theorem}\label{thm:kdnf}
  The models of a $k$-DNF with $n$ variables can be enumerated with precomputation in $O(n)$ and $O(k^{3/2}2^{2k})$ delay.
\end{theorem}


 To explain our algorithm, we need first to introduce notations. Let $D$ be a DNF-formula on variables $X$. Assume without loss of generality that $X$ is ordered with $<$. Given a term $T \in D$, we denote by $\one_T : \var(T) \rightarrow \{0,1\}$ the only model of $T$, that is, for every $\ell \in T$, $\one_T(x) = 1$ if $\ell = x$ and $\one_T(x) = 0$ if $\ell = \neg x$.

If $y \in \var(T)$, we denote by $\zero^y_T : \{z \in \var(T) \mid z \leq y\} \rightarrow \{0,1\}$ the assignment defined by $\zero^y_T(z) = \one_T(z)$ for $z < y$ and $z \in \var(T)$ and $\zero^y_T(y) = 1-\one_T(y)$.

For example, if $T = x_1 \land x_2 \land \neg x_3$, we have $\one_T = \{x_1\mapsto 1, x_2 \mapsto 1, x_3 \mapsto 0\}$ and $\zero_T^{x_1} = \{x_1 \mapsto 0\}$, $\zero_T^{x_2} = \{x_1 \mapsto 1, x_2 \mapsto 0\}$ and $\zero_T^{x_3} = \{x_1 \mapsto 1, x_2 \mapsto 1, x_3 \mapsto 1\}$.

These assignments naturally induce a partitioning of the model of a DNF:

\begin{lemma}
\label{lem:partition}
  Given a DNF $D$ and a term $T \in D$, we have:
  \[ \sat(D) = \sat(D, \one_T) \uplus \biguplus_{y \in \var(T)} \sat(D,\zero_T^y). \]
\end{lemma}
\begin{proof}
  The right-to-left inclusion is clear as $\sat(D,\tau) \subseteq \sat(D)$ for
  any $\tau$. Moreover, these unions are clearly disjoint since for every $y,z
  \in \var(T)$, $z < y$, we have $\zero^z_T(z) = 1-\one_T(z)$ and $\zero^y_T(z) =
  \one_T(z)$ that is $\zero^z_T(z) \neq \zero^y_T(z)$. Similarly, $\zero_T^y(y) = 1-\one_T(y)$, that is $\one_T(y) \neq \zero_T^y(y)$ by definition.

  For the left-to-right inclusion, let $\tau \in \sat(D)$. If $\tau \simeq \one_T$, then $\tau \in \sat(D, \one_T)$. Otherwise, let $y$ be the smallest variable of $\var(T)$ such that $\tau(y) \neq \one_T(y)$. Then we have $\tau \simeq \zero_T^y$.
\end{proof}

\begin{proof}[Proof (of Theorem~\ref{thm:kdnf}).]

Given a $k$-DNF $D$ on variables $X$, we use Lemma~\ref{lem:partition} to enumerate $\sat(D)$. We denote by $X = \var(D)$,  $N = |X|$ and  $M = |D|$. 

We start by picking one term $T \in D$. We choose it to have a minimal number of
literals. By definition, it contains at most $k$ literals. Enumerating $\sat(D,
\one_T)$ corresponds to enumerating the models of $T$ and can be done with
$O(1)$ delay by Proposition~\ref{prop:Gray}. Observe that the precomputation
time boils down to choosing a term and outputting its first solution which can
obviously be done in time $O(n)$ as stated in the theorem.

Between the output of two solutions of $\sat(D, \one_T)$, we spend some extra time to precompute $D[\zero_T^y]$ for every $y \in \var(T)$. To store the formula $D$ and the subformulas 
we will need, we use a trie of terms as explained in Section~\ref{par:trie}.
From the trie of $D$ we can compute the trie of $D[\zero_T^y]$ in time $O(kM)$, where $M$
is the number of terms of $D$ by traversing $D$, detecting the occurrences of the literals which should be set to $0$ or $1$. When a path contains a literal set to one, we contract it by removing this node and connecting its parents to its children. When a path contains a literal set to zero, we remove the subtree and all ancestors of degree one since they do not represent a term anymore.
In the end, we need $O(k^2M)$ steps to precompute $D[\zero_T^y]$ for every $y \in \var(T)$. Let $A$ be a constant such that this precomputation can be done in time $A \cdot k^2M$ steps at most.
 
Assume that between the output of two solutions of $\sat(D, \one_T)$, we allow  $dA$ steps for this precomputation (the value of $d$ will be fixed later depending on our needs). Since $|\sat(D, \one_T)| \geq 2^{N-k}$, this gives us a total amount of $2^{N-k}dA$ steps for this precomputation. Thus, if $2^{N-k}dA > Ak^2M$, that is, if
\begin{equation} \label{eq:steps}
  2^{N-k}d \geq k^2M
\end{equation}
we have enough time to compute $D[\zero_T^y]$ for every $y \in \var(T)$. If this is the case, then we do the precomputation, finish the enumeration of $\sat(D,\one_T)$ and then recursively start the enumeration of $\zero_T^y \times D[\zero_T^y]$ for each $y \in \var(T)$. The number of variables of $D[\zero_T^y]$ has of course decreased but we can still allow for $dA$ steps of extra computation between the output of two solutions.

To make sure the algorithm works, we have to ensure that~(\ref{eq:steps}) is
true even during a recursive call where the DNF $D'$ we are using to enumerate
the solutions of $D$ may have less variables and less terms. To ensure this, we
will pick $d$ sufficiently large. We now show that picking $d = k^{3/2} 2^{2k}$
is sufficient.

Let $D'$ be a DNF that is used at some point in a recursive call to enumerate
the solutions of $D$. We denote by $X' = \var(D')$, $N' = |X'|$ and $M' = |D'|$.
We follow the previous strategy. We pick a term $T'$ of $D'$ of smallest
size. We denote by $k' \leq k$ the number of literals it contains. By
definition, every terms in $D'$ have at least $k'$ and at most $k$ literals.

As long as
\begin{equation} \label{eq:stepsp}
  2^{N'-k'}d \geq k'^2M'
\end{equation}
is true, the strategy described above will work and ensure that we can
enumerate the solutions of $D'$ with delay $dA$ while preparing the inputs of
the next recursive calls. We thus have to check that we always have:

\begin{equation} \label{eq:border}
  d = k^{3/2} 2^{2k} \geq k'^2M' \times 2^{k'-N'}
\end{equation}

Let $M(k',k)$ be the maximal number of terms on $N'$ variables containing at
least $k'$ and at most $k$ literals. By definition, $M' \leq M(k,k')$. We start
by proving that by induction that $2^{k'}M(k',k) \leq 2^kM(k,k)$. If $k=k'$ then
the equality is clear. Otherwise, if $k' < k$:

\begin{align*}
  2^{k'}M(k',k) & = 2^{k'}M(k',k')+{1 \over 2}2^{k'+1}M(k'+1,k) &\text{by definition of $M(k',k)$}\\
                & \leq 2^{k'}M(k',k')+2^{k-1}M(k,k) &\text{by induction} \\
                & \leq 2^{k}M(k,k) & \text{since $M(k',k') < M(k,k)$.}
\end{align*}

There are exactly $2^k{N' \choose k}$ terms having exactly $k$ literals on
$N'$ variables. It follows that

\begin{align*}
  k'^2 \times 2^{k'}M' \times 2^{-N'} & \leq k^22^{2k-N'}{N' \choose k}  & \text{from what precedes,} \\
                                      & \leq k^22^{2k-2k}{2k \choose k} & \text{since $2^{-N'}{N' \choose k}$ is maximal for $2k=N'$,} \\
                                      & \leq k^2{(2k)!\over (k!)^2} & \text{by definition,} \\
                                      & \leq k^2{\sqrt{4k\pi} ({2k \over e})^{2k} \over {2k\pi ({k \over e})^{2k}}} & \text{by Stirling's formula, } n! = \sqrt{2n\pi}\big({n \over e}\big)^n+o(1) \\
           & \leq k^2{2^{2k} \over \sqrt{\pi k}} & \text{after simplifying,} \\
           & \leq k^{3/2}2^{2k} = d.
\end{align*}
which is the desired inequality.
\end{proof}

\begin{algorithm}
  
  \KwData{A $k$-DNF-formula $D$}
  \Begin{
    \If{$D = \emptyset$}{\Return $\emptyset$ \;}
    \eIf{$D = \{C\}$}
    {Enumerates the models of $C$ \;}
    {
      $d \leftarrow Ak^{3/2}2^{2k}$ \;
      Pick $C \in D$ of minimal size\;
      Every $d$ steps of computation in the next block, output a new model of $D[\one_C]$ \;
      \Begin{
      \For{$y \in \var(C)$}{
        $D^y \leftarrow D[\zero_C^y]$ \;
      }
    }
    \For{$y \in \var(C)$}{
      Recursively enumerates $\zero_C^y \times \sat(D^y)$\;
    }
    }
    }
  \caption{Enumerates the models of $k$-DNF with delay $2^{O(k)}$.\label{alg:kdnfenum}}
\end{algorithm}

  A pseudo code for the algorithm of Theorem~\ref{thm:kdnf} is given in Figure~\ref{alg:kdnfenum}. It has \emph{constant delay} for constant $k$ and it is in {\strongpdelay} (polynomial in $n$ the size of a solution) for terms of size $O(\log(n))$.
A natural question is then to further improve the delay for $k$-DNF. An algorithm with polynomial delay in $k$ would be a {\strongpdelay} algorithm for the general case. Already a subexponential delay in $k$ would be an improvement for the general case, at least for formulas with many terms.

We can use Algorithm~\ref{alg:kdnfenum} to significantly improve a result on the enumeration of models of first order formula with free second order variables. In~\cite{durand2011enumeration}, the class of 
$\Sigma_1$ formulas is defined as the set of first order formulas with with free second order variables, that can be written as a single block of existential quantifiers followed by a quantifier free formula (quantifier depth one). In~\cite{durand2011enumeration} Theorem~$10$, it is proved that the models of a $\Sigma_1$ formula can be enumerated in polynomial delay using a method similar to Proposition~\ref{prop:enum_union}. Moreover, this problem is shown to be equivalent to the enumeration of models of a $k$-DNF. As a corollary of Theorem~\ref{thm:kdnf}, we obtain the following theorem.
\begin{theorem}
 The models of a $\Sigma_1$ formula can be enumerated in constant delay.
\end{theorem}

The complexity considered in the previous theorem is the \emph{data complexity}, that is the size of the formula is considered fixed and the complexity depends only on the size of the model/database the formula is evaluated on. However, the dependency in the size of the formula is only singly exponential, which is low in this context. The result is surprising since we obtain the same complexity as for $\Sigma_0$ (quantifier free formulas), the first level of the hierarchy which is far simpler and for which we use a different algorithm. However, there is one difference between $\Sigma_0$ and $\Sigma_1$: the first can be solved in constant time in a more restricted 
machine model, where the set of output registers is always the same.

\section{Average delay of enumerating the models of DNF}\label{sec:average}

    \newcommand{\nox}[1]{\underline{#1}}

In this section, we analyze the average delay of the Flashlight method, using appropriate data structures, to show it is better than the delay. The idea is to amortize the cost of maintaining the formula $D[\tau]$ during the  traversal of the tree over all models of $D[\tau]$ in the spirit of~\cite{uno2015constant}.
 To do that we exhibit a relation between the number of models of a DNF and its number of terms.
 For any class of DNF for which it is possible to guarantee, for all partial assignments $\tau$, that the number of models of $D[\tau]$ is large enough with regard to $|D[\tau]|$, we obtain a good average delay.
We then improve the delay of the Flashlight method by improving the computation of the subformulas when branching and showing that either the number of terms decreases fast which makes the depth small or the cost of branching is small.

  From now on, let $\gamma = \log_3(2)$. 
  If we consider the DNF with all possible terms on $n$ variables, it has $2^n$ models and $3^n$ terms.
  Hence, there is a DNF formula with $m$ non empty distinct terms and $m^{\gamma}$ models
  and we are able to prove a matching lower bound on the number of terms.
  
  \begin{lemma}\label{lem:number_general}
   A DNF formula with $m$ non empty distinct terms has at least $m^{\gamma}$ models.
  \end{lemma}

  \begin{proof}

   The proof of the lemma is by induction on $n$ the number of variables in the DNF. 
   For $n=1$, the formula may have zero term and zero solution, one term and one solution or two terms
   and two solutions. All these cases satisfy the property. 
   Assume the lemma proved for $n$ variables and consider $D$ a DNF with $n+1$ variables.
   Let $x$ be some variable of $D$, we want to project out $x$ to apply the induction hypothesis.
   Let us define a partition of $D$ into three formulas over $\var(D)\setminus\{x\}$: 
   \begin{itemize}
    \item $D_x$ is the set of terms of $D$ containing the literal $x$, with $x$ removed; 
    \item $D_{\bar{x}}$ is the set of terms of $D$ containing the literal $\bar{x}$, with $\bar{x}$ removed; 
    \item $D_{\nox{x}}$ is the set of terms of $D$ containing neither $x$ or $\bar{x}$.
   \end{itemize}

   We let $s$ be the number of models of $D$, $s_x$ the number of models of
   $D_x$, $s_{\bar{x}}$ the number of models of $D_{\bar{x}}$ and $s_{\nox{x}}$
   the number of $D_{\nox{x}}$. We relate $s$ to $s_x$, $s_{\bar{x}}$ and
   $s_{\nox{x}}$ the number of models of $D_x$, $D_{\bar{x}}$ and $D_{\nox{x}}$.
   All models of $D_x$ extended by $x=1$ are models of $D$ as are the models of
   $D_{\bar{x}}$ extended by $x=0$. Moreover these models are distinct. Now
   consider the models of $D_{\nox{x}}$, they can be extended by $x=0$ and $x=1$
   to give a model of $D$. We decompose $s_{\nox{x}}$: $s'_{\nox{x}}$ is the
   number of models of $D_{\nox{x}}$ which are not model of both $D_x$ and
   $D_{\bar{x}}$ while $s_r$ is the number of models which satisfy $D_x$,
   $D_{\bar{x}}$ and $D_{\nox{x}}$. We write $s_x = s'_x + s_r$ and $s_{\bar{x}}
   = s'_{\bar{x}} + s_r$. Remark that a solution satisfying $D_x$, $D_{\bar{x}}$
   and $D_{\nox{x}}$ yields two solutions of $D$. Hence we have the following
   equation:
   \begin{equation}\label{eq:decomposition}
   s \geq s'_x + s'_{\bar{x}} + s'_{\nox{x}} + 2s_r.
   \end{equation}

   It is not an equality, since the models counted by $s'_{\nox{x}}$ may contribute one or two models 
   to $D$. 
 
    By induction hypothesis, we have for the three formulas on $n$ variables: 
   \begin{itemize}
    \item $|D_x|^{\gamma} \leq s'_x + s_r$
    \item $|D_{\bar{x}}|^{\gamma} \leq s'_{\bar{x}} + s_r$
    \item $|D_{\nox{x}}|^{\gamma} \leq s'_{\nox{x}} + s_r$
    \end{itemize}
    
    By construction, the number of terms of $D$ is equal to $|D_x| + |D_{\bar{x}}| +
    |D_{\nox{x}}|$. Hence we have $|D|^{\gamma} = (|D_x| + |D_{\bar{x}}| + |D_{\nox{x}}|)^{\gamma}$.
    We use the inequalities from the induction hypothesis: $|D|^{\gamma} \leq ((s'_x + s_r)^{1/\gamma} +
    (s'_{\bar{x}} + s_r)^{1/\gamma} + (s'_{\nox{x}} + s_r)^{1/\gamma})^{\gamma}$.
    Since $1/\gamma > 1$, we use Minkowski inequality to obtain 
    $$|D|^{\gamma} \leq ((s'_{x})^{1/\gamma}  +
    (s'_{\bar{x}})^{1/\gamma} + (s'_{\nox{x}})^{1/\gamma})^{\gamma} + (3(s_r)^{1/\gamma})^{\gamma}.$$
    Since $\gamma < 1$, the function $x\rightarrow x^\gamma$ is concave and we have the following inequality:
    $$((s'_{x})^{1/\gamma}  + (s'_{\bar{x}})^{1/\gamma} + (s'_{\nox{x}})^{1/\gamma})^{\gamma} \leq s'_x + s'_{\bar{x}} + s'_{\nox{x}}.$$
    Finally, we use the equality $3^\gamma =2$ and Eq~\ref{eq:decomposition} to obtain 
    $$|D|^{\gamma} \leq s'_x + s'_{\bar{x}} + s'_{\nox{x}} + 2s_r \leq s$$
    which establishes the property for $D$.
  \end{proof}

  \paragraph{Amortized analysis of flashlight algorithms} In
  theorems~\ref{th:average_delay} and \ref{thm:best_average_delay}, we analyze
  the average delay of flashlight algorithms using amortization. Both algorithms
  boil down to exploring the tree $T_D$ defined in
  Section~\ref{prop:classical_flashlight}. For each node $\tau$ of $T_D$, we
  will start by evaluating the computation time $c(\tau)$ we need in the
  algorithm each time we visit this node. We will refer to this quantity as the
  \emph{cost} of $\tau$. Then we divide (not necessarily uniformly) this cost
  among all models of $D$ that are under $\tau$ in $T_D$. The fraction of the
  cost of $c(\tau)$ assigned by $\tau$ on models $\sigma$ is called the
  \emph{charge} of $\tau$ on $\sigma$. The important invariant that we need to
  maintain is that the sum of the charges of $\tau$ on $\sigma$ over all
  models $\sigma$ below $\tau$ in $T_D$ has to be greater than the cost of
  $\tau$. The \emph{charge} of a model $\sigma$ of $D$ is the sum over every
  ancestor $\tau$ of $\sigma$ in $T_D$ of the charge of $\tau$ on $\sigma$. By
  definition, the sum of the charges of all models of $D$ is greater than the
  overall complexity of the algorithm. Thus, an upper bound on the maximal
  charge a model of $D$ can have is also an upper bound on the average delay
  of the algorithm.

  To improve the average delay of the flashlight method, we need to use an adapted data structure to represent the formula. In particular we need that, when considering some inner node $D[\tau]$ of $T_D$, the cost to process it depends on $|D[\tau]|$ instead on $|D|$. In particular, we need to guarantee that there are no redundancy of terms in the structure representing $D[\tau]$ and that we can maintain it efficiently, that is why we again use a trie of terms to represent $D$.
  
  \begin{theorem}\label{th:average_delay}
    The models of a DNF can be enumerated with average delay $O(n^2m^{1-\gamma})$ and polynomial space.
  \end{theorem}
\begin{proof}

  We maintain the formula $D[\tau]$ when we traverse $T_D$ using the trie
  containing its terms as explained in Section~\ref{par:trie}. On a node $\tau$
  we can decide quickly whether $D[\tau]$ has a model and we can maintain
  $D[\tau]$ efficiently without redundancy of terms. In the flashlight search,
  we will fix the variables following their order $x_1,\dots,x_k$. Hence,
  visiting a child $\tau'$ of $\tau$ corresponds to setting the value of some
  variable $x_k$ with all variables $x_i$ with $i < k$ already fixed. If we set
  $x_k$ to $0$ then we need to remove the subtree under the root of the trie,
  with first node $x_k$. Then we remove the subtree under the root of the trie,
  with first node $\bar{x}_k$, and insert back all elements in this subtree into
  the trie without the first node $\bar{x}_k$. The complexity of the latter is
  $O(n \cdot |D[\tau]|)$ since the number of terms is bounded by $|D[\tau]|$ (no
  term appears several times in the trie). To set $x_k$ to $1$, we do the same
  operation where we exchange the roles of $x_k$ and $\bar{x}_k$. To be able to
  go up in the tree during the flashlight search, we must restore the trie to
  its previous state. To do that in time $O(n \cdot |D[\tau]|)$, it is enough to
  store the list of elements which have been removed or added in the trie when
  going down the same edge and to reverse the operations. During the
  enumeration, the node $\tau$ will be visited at most four times: twice going
  down, setting first $x_k$ to $0$ and then to $1$, and twice when going up in
  the tree. Thus, the cost $c(\tau)$ of node $\tau$ is $O(n \cdot |D[\tau]|)$.
  The additional memory used during the algorithm is bounded by $O(mn^2)$.

  We now compute the average delay by charging the models under $\tau$ in $T_D$
  as explained before the statement of this theorem. We distribute the cost of
  $\tau$ uniformly on every solution $\sigma$ under $\tau$, that is, if there
  are $S$ models under $\tau$ in $T_D$, the charge of $\tau$ on each of these
  models $\sigma$ is $c(\tau) \over S$. By Lemma~\ref{lem:number_general}, we
  have that $D[\tau]$ has at least $|D[\tau]|^{\gamma}$ models thus the charge
  of $\tau$ on a model $\sigma$ is at most $n |D[\tau]|/|D[\tau]|^\gamma \leq
  nm^{1-\gamma}$. The charge of an ancestor of a model is thus at most
  $nm^{1-\gamma}$. Since it has at most $n$ ancestors, the average delay is
  bounded by $O(n^2m^{1-\gamma})$.
\end{proof}
%
%
%

We now improve both the previous algorithm and the analysis of its delay. In particular, we better use the fact that the sizes of the considered formulas are decreasing when going down the tree of the flashlight method. It shows that this method is always better on average than the original flashlight search we have presented in Proposition~\ref{prop:classical_flashlight}.

  \begin{theorem}\label{thm:best_average_delay}
    The models of a DNF can be enumerated with average delay $O(nm^{1-\gamma})$ and polynomial space.
  \end{theorem}
\begin{proof}
  Let $D$ be a DNF formula. Assume that we are exploring $T_D$ at
  node $\tau$ in the flashlight method. To ease notation in what follows, we
  denote $D[\tau]$ by $D'$. As before, the exploration of the node $\tau$ boils
  down to computing the two subformulas $D'_0 := D'[x \rightarrow 0]$ and $D'_1
  := D'[x \rightarrow 1]$ from $D'$. We denote their respective size by $m_0$
  and $m_1$. Without loss of generality, we assume that $m_0 \leq m_1$. To improve the complexity
  of our algorithm, we will construct $D'_0$ and $D'_1$ differently depending on
  the value of $m_0$. We consider two cases.

  \textbf{Slow branching:} Assume that $m_0 \geq m/2$, we say that we are in the
  slow case because the size of the subformulas does not decrease much. We divide
  the terms of $D'$ into $D'_{x}$ (terms with $x$), $D'_{\bar{x}}$ (terms with
  $\bar{x}$) and $D'_{\nox{x}}$ (terms without the variable $x$). To compute
  $D'_0$, we need to remove $D'_{x}$ from $D$. In the trie, it is the subtree
  under the root with label $x$. It can be done in time $O(1)$ and it can be
  later restored in time $O(1)$ as it is a single pointer to change. We need
  also to remove the $\bar{x}$ for all terms in $D'_{\bar{x}}$. It is done by
  inserting back in the trie all elements of its subtree with label $\bar{x}$
  without their first node. Remark that there are $|D'_{\bar{x}}|$ such elements
  which is bounded by $m - m_1$. To be able to restore the trie to its original
  state, we keep a list of the elements inserted so that we can easily remove
  them and insert them back in the trie with their first literal $\bar{x}$. The
  complexity to compute $D'_0$ and then to restore $D'$ is in $O((m - m_1)n)$.
  Similarly, we compute $D'_1$ in time $O((m-m_0)n)$. Observe that until now, we
  proceed as in the proof of Theorem~\ref{th:average_delay}. As before, the node
  $\tau$ is visited at most $4$ times, thus, the overall cost of $\tau$ is
  $O((m-m_0)n)+O((m-m_1)n)$.

  Each model of $D'$ receives a charge from $\tau$ as follows: we distribute the
  cost of building $D_0'$ uniformly over the models of $D_1'$. That is, if
  $D_0'$ has $S$ models, each model of $D_1'$ receives a charge of
  $O(n(m-m_1)/S)$. By hypothesis, $D'_0$ and $D'_1$ have at least $m/2$ terms
  each, hence by Lemma~\ref{lem:number_general} they have each at least
  $(m/2)^{\gamma}$ models. Thus, each model of $D_1'$ receives a charge of at
  most $n(m-m_1)/m^{\gamma}$. Similarly, we distribute the cost of building
  $D_1'$ uniformly over the models of $D_0'$. As before, each model of $D_0'$
  receives a charge of at most $n(m-m_0)/m^{\gamma}$.


  \textbf{Fast branching:} Assume that $m_0 < m/2$, we say that we are in the
  fast case because the size of at least one subformula decreases by at least a
  factor of two. We charge the models of this formula only. First remark
  that we can build the trie representing $D'_0$ in time $nm_0$ as in the
  previous case. For $D'_1$, we change one thing. Instead of inserting the terms
  of $D_x$ in the trie without their first node, we take the subtree of terms
  beginning by $x$ as our base trie and insert in it the terms of
  $D'_{\nox{x}}$. The complexity of these operations is now $O(n \cdot
  |D'_{\nox{x}}|)$ instead of $O(n\cdot |D'_{x}|)$. Since $D'_{\nox{x}}$ is a
  subset of $D'_0$, it has less than $m_0$ terms. As a consequence, the cost of
  $\tau$ is $O(nm_0)$. We distribute uniformly this cost over the models of
  $D'_0$. Since $D_0'$ has at least $m_0^{\gamma}$, each model of $D_0'$
  receives a charge of $nm_{0}^{1-\gamma}$.
  We do not need to know whether $m_0 < m/2$ to use this method. It is enough to 
  know the relative size of $D'_{\nox{x}}$, $D'_{x}$ and $D'_{\bar{x}}$, which are easy to maintain in the trie, and to always insert the small set into the large one.

  \textbf{Analysis of the average delay.} Now we can evaluate the average delay
  by bounding the overall charge of each model of $D$. We fix a model $\sigma$
  of $D$. We need to give an upper bound on the sum of the charges received from
  all its ancestors.

  We first evaluate the charge received from ancestor nodes corresponding to
  slow branching. We first look at the charge received from slow branching
  ancestor whose branching results in a formula having more than $m/2$ terms.
  Let $a_1,\dots, a_l$ be the sequence of these ancestors and let $m \geq m_1 \geq
  \dots \geq m_l \geq m/2$ denote the decreasing sequence of the number of terms
  of the corresponding formulas. Observe that if there is a fast branching at
  node $\tau$ in between these slow branchings, then it does not charge
  $\sigma$. Indeed, fast branching only charges the branch whose number of terms
  is divided by $2$, which would result in a number of terms below the $m/2$
  threshold.

  The charge of node $a_{i}$ is thus at most $n(m_i-m_{i+1})/m_i^{\gamma}$. By
  summing all charges, we get an overall charge of $n\sum_{i=1}^{l-1}
  (m_i-m_{i+1})/m_i^{\gamma}$. Since $m_i \geq m/2$ for all $i$, this overall
  charge is at most $n(m/2)^{-\gamma} \sum_{i=1}^{l-1} (m_i-m_{i+1}) \leq
  n(m/2)^{1 - \gamma}$.

  Using the same proof for slow branching ancestors whose number of terms is
  between $m/2^j$ and $m/2^{j+1}$, we obtain a bound of $n(m/2^j)^{1 - \gamma}$ 
  for the overall charge received by these ancestors. When we sum
  all these charges, we have a geometric sum which yields a $O(nm^{1-\gamma})$
  upper bound on the charges received by all slow branching ancestors.

  Let us consider the fast nodes on the path to the root, their sizes are
  denoted by $m'_1,\dots,m'_t$. The overall charge received by all fast
  branching ancestors is $\sum_{i=1}^{t} nm'^{1-\gamma}_{i}$. Since for all $i
  \leq t$, $m'_{i+1} < m'_i/2$, we can bound the cost by the geometric sum
  $\sum_{i=1}^{\infty} n\left(m/2^i\right)^{1-\gamma} \leq 2nm^{1-\gamma}$.
  The average delay is thus $O(nm^{1-\gamma})$.
\end{proof}

\paragraph{Back to $k$-DNF}
In the last theorem, the dependency of the average delay in $n$ comes only from the size of terms stored 
in the trie. Hence, if they are smaller, say of size $k$, the $n$ of the delay becomes a $k$. 
We now show how to modify Algorithm~\ref{alg:kdnfenum} for $k$-DNF, to get the best possible average delay. The idea is to use it until the number of variables is small, i.e. $\lambda k$ with $\lambda$ constant but larger than two and at this point use the algorithm of Theorem~\ref{thm:best_average_delay}.  

\begin{theorem}
  \label{th:kdnfamortized}
 There is an algorithm with average delay $O(2^{3k/2})$ to enumerate the models of a $k$-DNF.
\end{theorem}
\begin{proof}

  The algorithm is a mix between Alogrithm~\ref{alg:kdnfenum} and the algorithm
  of Theorem~\ref{thm:best_average_delay}. More precisely, let $\lambda \in
  \R_+$ be a constant whose value will be fixed later. If the number of
  variables of the DNF is greater than $\lambda k$, we enumerate the model of
  the DNF with the same strategy as in Algorithm~\ref{alg:kdnfenum}. If, in a
  recursive call on a DNF $D'$, the number of variables of $D'$ is below
  $\lambda k$, we enumerate all models of $D'$ using the algorithm of
  Theorem~\ref{thm:best_average_delay}.

  To precisely analyse the average complexity of this algorithm, we need a bound
  on the number $m(n,k)$ of terms having $k$ literals on $n$ variables. It is
  easy to see that $m(n,k) = \sum_{i=0}^k {n \choose i} 2^i$ as one can pick a
  term by first choosing its variables and then their sign. Let $H(p) = -p
  \log(p) - (1-p)\log(1-p)$ be the binary entropy. It is a well-known
  inequality~\cite{flum2006parameterized} that $\displaystyle{\sum_{i = 0}^{ k}
    {n \choose i}} \leq 2^{nH(k/n)}$. Hence,

  \[m(n,k) \leq 2^k \sum_{i = 0}^{ k} {n \choose i} \leq 2^{nH(k/n)+k}.\]

  When we use the algorithm of Theorem~\ref{thm:best_average_delay} with average
  delay $O(n m^{1-\gamma})$, we have less than $\lambda k$ variables and,
  consequently, less than $m(\lambda k, k)$ terms. We thus have an average delay
  of $O(\lambda km^{1-\gamma}(n,k)) = O(\lambda k 2^{(1-\gamma)(\lambda k H(1/\lambda)+k)})$.

  When using Algorithm~\ref{alg:kdnfenum}, Equation~\ref{eq:steps}, guarantees a
  delay $k^2m(n,k)2^{k-n}$. Observe that this delay is decreasing as $n$
  increases as $n \mapsto 2^{-n}m(n,k)$ decreases. Thus, the worst delay in our
  algorithm is when $n$ is the smallest. By definition of our algorithm, it
  happens when $n = \lambda k$. We thus have a delay of at most $k^2 2^{\lambda
    k (H(1/\lambda)-1)+2k}$.

  To optimize the complexity of the general algorithm we need to determine when
  both complexity are equal since the delay of the first stage increases when
  $\lambda$ decreases, while the average delay of the second stage decreases
  with $\lambda$. We neglect the factor of $\lambda$, $k$ and $k^2$ in the
  algorithms to simplify the computation, while increasing the complexity of at
  most $\lambda k$. The constant $\lambda$ must satisfy:

  \begin{align*}
    \lambda k (H(1/\lambda)-1)+2k = & (1-\gamma)(\lambda k H(1/\lambda)+k)
  \end{align*}

  which is equivalent, after simplifications, to \[1-\lambda+\gamma +
    \gamma(\lambda \log(\lambda) - (\lambda-1)\log(\lambda-1)) = 0.\] A
  solution of this equation, obtained with a numerical computing software, is
  $\lambda \simeq 3.55301$ which corresponds to an average delay a bit better
  than $2^{3k/2}$.

\end{proof}

\section{Enumerating models of monotone DNF}
\label{sec:monotone}

When the underlying formula is monotone, that is it does not contain any negated literal, we can enumerate the models with a delay polynomial in the number of variables only. However, our current techniques need exponential memory to work.

\begin{theorem}
  \label{thm:montonednf}
  There is an algorithm that given a monotone DNF with $n$ variables and $m$
  terms, enumerates the models of $D$ with preprocessing $O(n m^2)$ and delay
  $O(n^2)$. The space needed for this algorithm is linear in the number of
  models of $D$. 
\end{theorem}
\begin{proof}
  We start by removing from $D$ every term $T'$ such that there exists $T \in D$
  with $T \subseteq T'$. Observe that it does not change the models of $T$ since
  $T' \Rightarrow T$. This preprocessing phase takes $O(nm^2)$ since comparing
  two terms may be done in time $O(n)$. Let $D'$ be the resulting minimized
  monotone DNF.

  The algorithm then work as follows: arbitrarily order the terms of $D' =
  \{T_1, \dots, T_p\}$ and its variables $X = \{x_1,\dots,x_n\}$. We initialize
  a trie $T$ that will contain the models that we have already enumerated,
  that is, in the following algorithm, each time we output a model, we store
  it in $T$, so we can check in time $O(n)$ if a model has already been
  enumerated.

  We now enumerate the models as follows: we start by enumerating all
  models of $T_1$. Then, we proceed by induction: once we have enumerated all
  models of $D'_i = T_1 \vee \dots \vee T_i$, we enumerate all models of
  $T_{i+1}$ that are not models of $D'_i$ until $i=p$. Once we are done, we
  have, by induction, enumerated all models of $D$.

  We claim that we can do this with delay $O(n^2)$ using a classical reverse
  search method. Let $Y=(y_1,\dots,y_m)$ be the variables that are not in
  $\var(T_i)$, ordered following the natural order we have on $X$. The models
  of $T_i$ are in one-to-one correspondence with $2^Y$ as follows: given $S
  \subseteq Y$, we have a model $m_S$ defined as $m_S(x) = 1$ if $x \in
  \var(T_i) \cup S$ and $m(x) = 0$ otherwise. We explore the models of $T_i$
  by following a tree $A$ whose nodes are labeled with $S$ for every $S
  \subseteq Y$. The root of the tree is labeled by $\emptyset$ and for every
  $S$, the unique predecessor of node $S$ is $S'$ with $S' = S \setminus
  \max(S)$. In other words, given $S$, the successors of $S$ are ${S \cup
    \{x_k\}}$ for $k > \max(S)$.

  We enumerate the models by following the structure of $A$. Each time we
  enumerate a model, we add it in the trie so that we can avoid to enumerate the
  same model twice.

  We start from the root of $A$, that is, we enumerate $m_\emptyset$. We claim
  that $m_\emptyset$ has not yet been enumerated. Indeed, assume toward a
  contradiction that $m_\emptyset \models T_j$ for $j < i$. Then since $D'$ is
  monotone, it means that $T_j \subseteq T_i$ which is absurd since $D'$ is
  minimized. Thus, the first model we enumerate is guaranteed to be fresh.

  Now we follow the structure of $A$ by depth first search. To make the
  exploration deterministic, we order the successors of a node $S$ as follows:
  $\{x_i\} \cup S < \{x_j\} \cup S$ if $i < j$. We maintain the following
  invariant: each time we visit a node $S$, we are guaranteed that model $m_S$
  has not yet been output. Observe that the invariant is true at the root of
  $A$.

  We also keep a pointer at $S$ to the next node $S'$ in $A$ that is not in the
  subtree rooted in $S$, has not yet been visited and such that $m_{S'}$ has not
  yet been enumerated. If no such node exists, the pointer is null and we will
  use it to detect that every model of the current terms has been output.

  The algorithm works as follows: when we visit node $S$, by the invariant, we
  are guaranteed that $m_S$ has not yet been enumerated. We start by outputting
  $m_S$. Then, we explore all its successors $S'$ in
  order (this step is only exploratory and do not count as ``visiting $S'$'' in
  our invariant) and test in time $O(n)$, by looking in the trie $T$, whether
  $m_{S'}$ has been already enumerated. The overall cost of this exploration is
  $O(n^2)$ since $S$ has at most $n$ successors. If for some successor $S'$ of
  $S$, $m_{S'}$ has already been output, then we can discard the whole
  subtree of $A$ rooted in $S'$. Indeed, by definition, every model $m_{S''}$
  in this subtree verifies $S'' \supseteq S'$. Thus, if $m_S'$ has been
  enumerated then $m_{S'} \models D'_{i-1}$ and since $D'_{i-1}$ is monotone,
  $m_{S''} \models D'_{i-1}$, so $m_{S''}$ has already been enumerated.

  If every successor of $S'$ has been discarded, then we follow the pointer in
  $S$ to the next node whose model has not yet been enumerated. If this pointer
  is null, then, by the invariant, we have enumerated every model of $D'_i$ and
  we can start enumerating the models of $T_{i+1}$.

  Otherwise, we start by creating the pointers for each successor of $S$ that
  contains fresh models. Let $S_1',\dots,S_r'$ be these successors in order. For
  $i < r$, we add a pointer from $S_i'$ to $S'_{i+1}$. We add a pointer from
  $S_r$ to $S''$ where $S''$ is the node pointed by $S$. This can be done in
  $O(n)$ as $r \leq n$. By construction, these new pointers verify the invariant
  that we are maintaining. We now recursively visit node $S_1'$.

  By construction of the invariant, we are guaranteed to output every model of
  $D'$ exactly once. The delay between the enumeration of two models is bounded
  by the time we need to insert the newly output model in the trie, pre-explore
  the successors of the current node and possibly follow a pointer to the next
  model. We thus have a delay of $O(n^2)$ which concludes the proof.
\end{proof}

While we do not rule out the existence of an improvement of the algorithm of
Theorem~\ref{thm:montonednf} to have a $O(n)$ delay by using a better data
structure than the trie to store the already enumerated models, improving the
space complexity of this algorithm seems much more challenging. Indeed, we
currently avoid repeating a model in the output by storing all enumerated
models. One could wonder whether we could have a more efficient procedure, that
does not depend on $i$, to decide whether a model of $T_{i+1}$ is already a
model of $T_1, \dots, T_i$. It boils down to testing whether a set of variables
is a subset of one $T_j$ for $j \leq i$. This problem is known as the
\emph{subset query problem}. In~\cite{williams2005new}, Williams observes all,
even non-trivial algorithms, for this problem either have an $\Omega(m)$ time
complexity (that is, we test whether the model is also a model of $T_j$ for
almost every $j \leq i$) or an $\Omega(2^n)$ space complexity, that is, we store
almost every possible model. It thus seems very difficult, and maybe impossible,
to improve the space complexity of the algorithm of Theorem~\ref{thm:montonednf}
significantly.

Finally, we observe that the algorithm described in Theorem~\ref{thm:montonednf}
also works if every variable appears only positively or negatively in all terms
of the formula. Indeed, if a variable $x$ appears only negatively, one can
replace every occurrence of $\bar{x}$ by $x$. This gives a monotone DNF whose
solutions are in one-to-one correspondence with the solutions of the original
DNF by simply swapping the value of the variables whose sign has been changed.
It can be done before outputting the solution.

\paragraph{Average Delay}

We now propose an algorithm with a good average delay for enumerating the models of a monotone DNF formula.
This allows to obtain a better \emph{average} delay than the previous theorem, while only using polynomial space. First, remark that monotone DNF formulas have at least as many models as terms: the function which to a term, associates the assignment with one on all the variables of the terms and zero elsewhere is an injection into the solutions of the formula. Hence, using the algorithm of Theorem~\ref{thm:best_average_delay} on a monotone formula, we obtain a better delay than in the general case, as stated in the next theorem.
  
  \begin{theorem}\label{th:amortized_monotone}
   The models of a monotone DNF can be enumerated with average delay $O(n)$ and polynomial space.
  \end{theorem}
 To improve the bound on the average delay using a similar algorithm, we should either guarantee a better relationship between the number of terms and the number of models or we should reduce the complexity of maintaining the trie during the algorithm.  Note that the formula with all positive terms has $2^n$ models but also $2^n$ terms, which show that the bound on the number of solutions we use is tight. If we further assume that no term are redundant, that is there are no $T_1,T_2$ such that $T_1 \subseteq T_2$, then the formula with all terms of size $n-1$ has $n + 1$ models and $n$ terms. Even when $m$ is large, the relationship is almost linear: the formula with all terms of size $n/2$ has $2^n/2$ models and $O(2^n/\sqrt{n})$ terms. 
 Hence, to improve the average delay, it seems hard to rely on a better bound on the number of models.
 However, the cost to deal with the trie can be reduced, as long as $m$ is not too large, by improving the way we store large terms in the trie. As a consequence, we obtain an algorithm that can be sublinear in the \textbf{model size}, when $m$ is not too large.
 
  \begin{theorem}\label{th:amortized_improved}
   The models of a monotone DNF can be enumerated with preprocessing $O(mn)$, average delay $O(\log(mn))$ and polynomial space.
  \end{theorem}
  \begin{proof}

    We use the algorithm of Theorem~\ref{thm:best_average_delay}, with a slight change in the datastructure used to represent $D[\tau]$, using the fact that it is monotone and we do not need to deal with the sign of the literals. Instead of representing each term of $D$ by its literals in order, we represent it by the  variables which do not appear in the term in order. That is a term is represented by the complementary term. The complementary terms are stored in a trie to represent $D$ and we have a $O(mn)$ preprocessing to compute this trie.

    During the algorithm, a variable may appear in all terms and thus do not appear in the trie of complementary terms. Such a variable must be set to one. We assume that the registers holding the value of each variable in the model to output are always set to $1$ when the variable is not yet fixed. Hence, we can forget about such a variable in the algorithm, since its value is always correct in the output. 

    Remark that when a term contains $n-k$ variables, it has $2^k$ models
    and the formula has at least as many models. In the algorithm of
    Theorem~\ref{thm:best_average_delay}, assume that $D[\tau]$
    has $n_{\tau}$ variables and there is a term with $n_{\tau}-k$ variables
    such that $k > \log(|D[\tau]|) + 2\log(n)$. Hence, $D[\tau]$ has at least
    $2^k$ models which is larger than $|D[\tau]|n^2$. We can charge the cost of
    of modifying $D[\tau]$ at node $\tau$ over all these models, hence each 
    model is charged $1/n$. The total cost charged to each model by all its ancestors of this type is thus bounded by $1$.

  Now assume that for some partial assignment $\tau$, all terms in $D[\tau]$ contain more than $n_{\tau}-k$ literals, with $k = \log(|D[\tau]|) + 2\log(n)$. Hence, in our representation of the trie by the variables not appearing in the term, each literal is represented by a path in the tree of size at most $k$.

  The complexity of the operation to get the representation of the subformulas are the same as with the tree of terms, and depends in the same way on the number of terms of the subformulas. However, it only modifies the trie of complementary terms of size $k$ instead of $n_{tau}$ and the number of models is as large as $m_{\tau}$ the number of terms instead of $m_{tau}^{\gamma}$. Hence, the charge of each model for these operations is bounded by $O(k)$, by the same analysis as in Theorem~\ref{thm:best_average_delay}.

  As a conclusion, the average complexity is in $O(k)$, that is in $O(\log(mn))$ which proves the proposition.
%
%
  \end{proof}

In~\cite{mary2016efficient}, we study the related problem of generating all unions of given subsets. An instance $\{s_1,\dots,s_m\}$ is a set of $m$ subsets of $[1,n]$ and we want to generate all distinct unions of these $s_i$. The delay of the algorithm in~\cite{mary2016efficient} is $O(nm)$, using a flashlight search algorithm similar to the algorithm of Proposition~\ref{prop:classical_flashlight}. Among  the enumeration problems captured by the framework of saturation by set operators~\cite{mary2016efficient} it is the only one not proved to be in {\strongpdelay} and it is also proved to be at least as hard to enumerate as the models of a monotone DNF.
We have also proved that enumerating the models of a monotone DNF is as least as hard as generating the unions~\cite{mary2016efficient}.

In the published version of this article~\cite{capelli2021enumerating}, we claim in Theorem $15$, without proof, that we can apply the method and analysis of Theorem~\ref{th:amortized_monotone} to obtain a $O(n)$ delay algorithm for the generation of unions of subsets. The main idea is that each distinct subset in the instance is also a solution, hence there are as least as many solutions as elements in the instance. However, this statement is \textbf{false}: the problem is that selecting an element to appear in a solution of the union problem changes the problem: it is not autoreducible. Hence, the argument that each element of the input yields a solution is not true anymore. 

The question of finding a {\strongpdelay} algorithm for generating the union of subsets is still an intriguing open problem.

 \subparagraph{Acknowledgements}

This work was partially supported by the French Agence Nationale de la
Recherche, AGGREG project reference ANR-14-CE25-0017-01. We thank Jan
Ramon for sharing his idea on the monotone DNF enumeration algorithm.

\bibliography{biblio}

\begin{thebibliography}{10}

\bibitem{AmarilliBJM17}
Antoine Amarilli, Pierre Bourhis, Louis Jachiet, and Stefan Mengel.
\newblock A circuit-based approach to efficient enumeration.
\newblock In {\em 44th International Colloquium on Automata, Languages, and
  Programming, {ICALP} 2017, July 10-14, 2017, Warsaw, Poland}, pages
  111:1--111:15, 2017.

\bibitem{andrade2016enumeration}
Ricardo Andrade, Martin Wannagat, Cecilia~C Klein, Vicente Acu{\~n}a, Alberto
  Marchetti-Spaccamela, Paulo~V Milreu, Leen Stougie, and Marie-France Sagot.
\newblock Enumeration of minimal stoichiometric precursor sets in metabolic
  networks.
\newblock {\em Algorithms for Molecular Biology}, 11(1):25, 2016.

\bibitem{avis1996reverse}
D.~Avis and K.~Fukuda.
\newblock {Reverse search for enumeration}.
\newblock {\em Discrete Applied Mathematics}, 65(1):21--46, 1996.

\bibitem{bagan2006mso}
Guillaume Bagan.
\newblock Mso queries on tree decomposable structures are computable with
  linear delay.
\newblock In {\em International Workshop on Computer Science Logic}, pages
  167--181. Springer, 2006.

\bibitem{barth2015efficient}
Dominique Barth, Olivier David, Franck Quessette, Vincent Reinhard, Yann
  Strozecki, and Sandrine Vial.
\newblock Efficient generation of stable planar cages for chemistry.
\newblock In {\em International Symposium on Experimental Algorithms}, pages
  235--246. Springer, 2015.

\bibitem{bohmova2018computing}
Kate{\v{r}}ina B{\"o}hmov{\'a}, Luca H{\"a}fliger, Mat{\'u}{\v{s}} Mihal{\'a}k,
  Tobias Pr{\"o}ger, Gustavo Sacomoto, and Marie-France Sagot.
\newblock Computing and listing st-paths in public transportation networks.
\newblock {\em Theory of Computing Systems}, 62(3):600--621, 2018.

\bibitem{capelli2019incremental}
Florent Capelli and Yann Strozecki.
\newblock Incremental delay enumeration: Space and time.
\newblock {\em Discrete Applied Mathematics}, 268:179--190, 2019.

\bibitem{capelli2021enumerating}
Florent Capelli and Yann Strozecki.
\newblock Enumerating models of dnf faster: breaking the dependency on the
  formula size.
\newblock {\em Discrete Applied Mathematics}, 303:203--215, 2021.

\bibitem{courcelle2009linear}
Bruno Courcelle.
\newblock Linear delay enumeration and monadic second-order logic.
\newblock {\em Discrete Applied Mathematics}, 157(12):2675--2700, 2009.

\bibitem{creignou1997generating}
Nadia Creignou and Jean-Jacques H{\'e}brard.
\newblock On generating all solutions of generalized satisfiability problems.
\newblock {\em Informatique th{\'e}orique et applications}, 31(6):499--511,
  1997.

\bibitem{durand2011enumeration}
Arnaud Durand and Yann Strozecki.
\newblock Enumeration complexity of logical query problems with second-order
  variables.
\newblock In {\em LIPIcs-Leibniz International Proceedings in Informatics},
  volume~12. Schloss Dagstuhl-Leibniz-Zentrum fuer Informatik, 2011.

\bibitem{florencio2015naive}
Christophe~Costa Flor{\^e}ncio, Jonny Daenen, Jan Ramon, Jan Van~den Bussche,
  and Dries Van~Dyck.
\newblock Naive infinite enumeration of context-free languages in incremental
  polynomial time.
\newblock {\em J. UCS}, 21(7):891--911, 2015.

\bibitem{flum2006parameterized}
J{\"o}rg Flum and Martin Grohe.
\newblock {\em Parameterized complexity theory}.
\newblock Springer Science \& Business Media, 2006.

\bibitem{fredkin1960trie}
Edward Fredkin.
\newblock Trie memory.
\newblock {\em Communications of the ACM}, 3(9):490--499, 1960.

\bibitem{golovach2018output}
Petr~A Golovach, Pinar Heggernes, Mamadou~Moustapha Kant{\'e}, Dieter Kratsch,
  Sigve~H S{\ae}ther, and Yngve Villanger.
\newblock Output-polynomial enumeration on graphs of bounded (local) linear
  mim-width.
\newblock {\em Algorithmica}, 80(2):714--741, 2018.

\bibitem{karp1983monte}
Richard~M Karp and Michael Luby.
\newblock Monte-carlo algorithms for enumeration and reliability problems.
\newblock In {\em 24th Annual Symposium on Foundations of Computer Science
  (sfcs 1983)}, pages 56--64. IEEE, 1983.

\bibitem{knuth2011combinatorial}
Donald~E Knuth.
\newblock Combinatorial algorithms, part 1, volume 4a of the art of computer
  programming, 2011.

\bibitem{knuth1997art}
Donald~Ervin Knuth.
\newblock {\em The art of computer programming}, volume~3.
\newblock Pearson Education, 1997.

\bibitem{mary2013enumeration}
Arnaud Mary.
\newblock {\em {\'E}num{\'e}ration des Dominants Minimaux d’un graphe}.
\newblock PhD thesis, Universit{\'e} Blaise Pascal, 2013.

\bibitem{mary2016efficient}
Arnaud Mary and Yann Strozecki.
\newblock Efficient enumeration of solutions produced by closure operations.
\newblock In {\em 33rd Symposium on Theoretical Aspects of Computer Science},
  2016.

\bibitem{dmtcs:5549}
Arnaud Mary and Yann Strozecki.
\newblock {Efficient enumeration of solutions produced by closure operations}.
\newblock {\em {Discrete Mathematics \& Theoretical Computer Science}}, {Vol.
  21 no. 3 }, 2019.

\bibitem{meel2018not}
Kuldeep~S Meel, Aditya~A Shrotri, and Moshe~Y Vardi.
\newblock Not all fprass are equal: demystifying fprass for dnf-counting.
\newblock {\em Constraints}, pages 1--23, 2018.

\bibitem{mehlhorn2013data}
Kurt Mehlhorn.
\newblock {\em Data structures and algorithms 1: Sorting and searching},
  volume~1.
\newblock Springer Science \& Business Media, 2013.

\bibitem{murakami2014efficient}
Keisuke Murakami and Takeaki Uno.
\newblock Efficient algorithms for dualizing large-scale hypergraphs.
\newblock {\em Discrete Applied Mathematics}, 170:83--94, 2014.

\bibitem{pruesse1994generating}
Gara Pruesse and Frank Ruskey.
\newblock Generating linear extensions fast.
\newblock {\em SIAM Journal on Computing}, 23(2):373--386, 1994.

\bibitem{read1975bounds}
RC~Read and RE~Tarjan.
\newblock {Bounds on backtrack algorithms for listing cycles, paths, and
  spanning trees}.
\newblock {\em Networks}, 5(3):237--252, 1975.

\bibitem{SchweikardtSV18}
Nicole Schweikardt, Luc Segoufin, and Alexandre Vigny.
\newblock Enumeration for fo queries over nowhere dense graphs.
\newblock In {\em Proceedings of the 37th ACM SIGMOD-SIGACT-SIGAI Symposium on
  Principles of Database Systems}, SIGMOD/PODS '18, pages 151--163, New York,
  NY, USA, 2018. ACM.
\newblock URL: \url{http://doi.acm.org/10.1145/3196959.3196971}, \href
  {http://dx.doi.org/10.1145/3196959.3196971}
  {\path{doi:10.1145/3196959.3196971}}.

\bibitem{segoufin2013enumerating}
Luc Segoufin.
\newblock Enumerating with constant delay the answers to a query.
\newblock In {\em Proceedings of the 16th International Conference on Database
  Theory}, pages 10--20. ACM, 2013.

\bibitem{strozecki2010enumeration}
Yann Strozecki.
\newblock {\em Enumeration complexity and matroid decomposition}.
\newblock PhD thesis, Paris 7, 2010.

\bibitem{strozecki2013enumerating}
Yann Strozecki.
\newblock On enumerating monomials and other combinatorial structures by
  polynomial interpolation.
\newblock {\em Theory of Computing Systems}, 53(4):532--568, 2013.

\bibitem{uno2015constant}
Takeaki Uno.
\newblock Constant time enumeration by amortization.
\newblock In {\em Workshop on Algorithms and Data Structures}, pages 593--605.
  Springer, 2015.

\bibitem{williams2005new}
Ryan Williams.
\newblock A new algorithm for optimal 2-constraint satisfaction and its
  implications.
\newblock {\em Theoretical Computer Science}, 348(2-3):357--365, 2005.

\bibitem{wright1986constant}
Robert~Alan Wright, Bruce Richmond, Andrew Odlyzko, and Brendan~D McKay.
\newblock Constant time generation of free trees.
\newblock {\em SIAM Journal on Computing}, 15(2):540--548, 1986.

\end{thebibliography}

\end{document}